\documentclass[11pt]{article}
\usepackage{graphicx}
\usepackage{subcaption}
\usepackage{todonotes}
\usepackage{mathtools}
\usepackage{soul}
\usepackage{amsthm}
\usepackage{amsmath}
\usepackage{systeme}
\usepackage{cite}

\DeclareMathOperator*{\argmin}{arg\,min}
\newtheorem{theorem}{Theorem}

\newtheorem{lemma}{Lemma}
\newtheorem{observation}{Observation}
\usepackage{fullpage}

\usepackage[pdftex, plainpages = false, pdfpagelabels,
                 bookmarks=false,
                 bookmarksopen = true,
                 bookmarksnumbered = true,
                 breaklinks = true,
                 linktocpage,
                 pagebackref,
                 colorlinks = true,
                 linkcolor = blue,
                 urlcolor  = cyan,
                 citecolor = red,
                 anchorcolor = green,
                 hyperindex = true,
                 hyperfigures
                 ]{hyperref}

\usepackage{amssymb}
\usepackage[ruled, noend, linesnumbered]{algorithm2e}
\newcommand{\R}{\mathbb{R}}
\newcommand{\dist}{\mathrm{dist}}

\def\calD{\mathcal{D}}
\def\vd{\mathcal{V}\!\mathcal{D}}

\title{An Improved Algorithm for Shortest Paths in Weighted Unit-Disk Graphs\thanks{A preliminary version of this paper will appear in {\em Proceedings of the 36th Canadian Conference on Computational Geometry (CCCG 2024)}. This research was supported in part by NSF under Grant CCF-2300356.}}
\author{
Bruce W. Brewer\thanks{Kahlert School of Computing,
University of Utah, Salt Lake City, UT 84112, USA. {\tt bruce.brewer@utah.edu}}
\and
Haitao Wang\thanks{Kahlert School of Computing,
University of Utah, Salt Lake City, UT 84112, USA. {\tt haitao.wang@utah.edu}}
}
\date{}

\begin{document}

\maketitle

\vspace{-0.2in}
\begin{abstract}
    Let $V$ be a set of $n$ points in the plane. The unit-disk graph $G = (V, E)$ has vertex set $V$ and an edge $e_{uv} \in E$ between vertices $u, v \in V$ if the Euclidean distance between $u$ and $v$ is at most 1. The weight of each edge $e_{uv}$ is the Euclidean distance between $u$ and $v$. Given $V$ and a source point $s\in V$, we consider the problem of computing shortest paths in $G$ from $s$ to all other vertices. The previously best algorithm for this problem runs in $O(n \log^2 n)$ time [Wang and Xue, SoCG'19]. The problem has an $\Omega(n\log n)$ lower bound under the algebraic decision tree model. In this paper, we present an improved algorithm of $O(n \log^2 n / \log \log n)$ time (under the standard real RAM model). Furthermore, we show that the problem can be solved using $O(n\log n)$ comparisons under the algebraic decision tree model, matching the $\Omega(n\log n)$ lower bound.
\end{abstract}


\section{Introduction}
\label{sec:intro}
Let $V$ be a set of $n$ points in the plane. The unit-disk graph $G = (V, E)$ has vertex set $V$ and an edge $e_{uv} \in E$ between vertices $u, v \in V$ if the Euclidean distance between $u$ and $v$ is at most 1. Alternatively, $G$ can be seen as the intersection graph of disks with radius $\frac{1}{2}$ centered at the points in $V$ (i.e., two disks have an edge in the graph if they intersect). In the {\em weighted graph}, the weight of each edge $e_{uv} \in E$ is the Euclidean distance between $u$ and $v$. In the {\em unweighted graph}, all edges have the same weight.

Given $V$ and a source point $s\in V$, we study the single source shortest path (SSSP) problem where the goal is to compute shortest paths from $s$ to all other vertices in $G$. Like in general graphs, the algorithm usually returns a shortest path tree rooted at $s$. The problem in the unweighted graph has an $\Omega(n\log n)$ lower bound in the algebraic decision tree model since even deciding if $G$ is connected requires that much time by a reduction from the max-gap~\cite{ref:CabelloSh15}. The unweighted problem has been solved optimally in $O(n \log n)$ time by Cabello and Jej\v ci\v c~\cite{ref:CabelloSh15}, or in $O(n)$ time by Chan and Skrepetos \cite{ref:ChanAl16} if the points of $V$ are pre-sorted (by both the $x$- and $y$-coordinates). Several algorithms for the weighted case are also known~\cite{ref:CabelloSh15,ref:KaplanDy20,ref:LiuNe22,ref:RodittyOn11,ref:WangNe20}. Roditty and Segal \cite{ref:RodittyOn11} first solved the problem in $(n^{4/3 + \delta})$ time, where $\delta > 0$ is an arbitrarily small constant. Cabello and Jej\v ci\v c \cite{ref:CabelloSh15} improved it to $O(n^{1 + \delta})$ time. Subsequent improvements were made by Kaplan, Mulzer, Roditty, Seiferth, and Sharir~\cite{ref:KaplanDy20} and also by Liu~\cite{ref:LiuNe22} by developing more efficient dynamic bichromatic closest pair data structures and plugging them into the algorithm of \cite{ref:CabelloSh15}. Wang and Xue~\cite{ref:WangNe20} proposed a new method that solves the problem in $O(n \log^2 n)$ time without using dynamic bichromatic closest pair data structures. It is currently the best algorithm for the problem.

\subsection{Our result}
We present a new algorithm of $O(n\log^2 n/\log\log n)$ time for the weighted case and, therefore, slightly improve the result of~\cite{ref:WangNe20}. Our algorithm follows the framework of Wang and Xue~\cite{ref:WangNe20} but provides a more efficient solution to a bottleneck subproblem in their algorithm, called the {\em offline insertion-only additively-weighted nearest neighbor problem with a separating line} (or IOAWNN-SL for short). Specifically, we are given a sequence of $n$ operations of the following two types: (1) Insertion: Insert a weighted point to $P$ (which is $\emptyset$ initially); (2) Query: given a query point $q$, find the {\em additively-weighted nearest neighbor} of $q$ in $P$, where the distance between $q$ to any point $p\in P$ is defined to be their Euclidean distance plus the weight of $p$. The points of $P$ and all the query points are required to be separated by a given line (say the $x$-axis). The goal of the problem is to answer all queries.

Wang and Xue~\cite{ref:WangNe20} solved the IOAWNN-SL problem in $O(n\log^2 n)$ time using the traditional logarithmic method of Bentley~\cite{ref:BentleyDe79}. This is the bottleneck of their overall shortest path algorithm; all other parts of the algorithm take $O(n\log n)$ time. We derive a more efficient algorithm that solves IOAWNN-SL in $O(n\log^2 n/\log\log n)$ time (see Theorem~\ref{thm:IOAWNN-SL} for details). Plugging this result into the algorithm framework of Wang and Xue~\cite{ref:WangNe20} solves the shortest path problem in $O(n\log^2 n/\log\log n)$ time.

\begin{theorem} \label{thm:IOAWNN-SL}
    Let $P$ be an initially empty set of $n$ weighted points in the plane such that all points of $P$ lie below the $x$-axis $\ell$. There exists a data structure $\calD(P)$ of $O(n)$ space supporting the following operations:
    \begin{enumerate}
        \item Insertion: Insert a weighted point $p$ below $\ell$ to $P$ in amortized $O(\log^2 n / \log \log n)$ time.
        \item Query: Given a query point $q$ above $\ell$, find the additively-weighted nearest neighbor to $q$ in $P$ in worst-case $O(\log^2 n / \log \log n)$ time.
    \end{enumerate}
\end{theorem}

Our algorithm for Theorem~\ref{thm:IOAWNN-SL} needs to solve a subproblem about merging two additively weighted Voronoi diagrams. Specifically, let $S_a$ and $S_b$ each be a subset of $n$ weighted points in the plane such that all points of $S_a\cup S_b$ are below the $x$-axis $\ell$. Let $\vd(S_a)$ denote the additively-weighted Voronoi diagram of $S_a$, and $\vd_+(S_a)$ denote the portion of $\vd(S_a)$ above $\ell$. Similarly, define $\vd(S_b)$ and $\vd_+(S_b)$ for $S_b$, and define $\vd(S_a\cup S_b)$ and $\vd_+(S_a\cup S_b)$ for $S_a\cup S_b$. Given $\vd_+(S_a)$ and $\vd_+(S_b)$, our problem is to compute $\vd_+(S_a\cup S_b)$. We solve the problem in $O(n)$ time by modifying Kirkpatrick's algorithm for merging two standard Voronoi diagrams~\cite{ref:KirkpatrickEf79} and by making use of the property that $\vd_+(S_a\cup S_b)$ and all points of $S_a\cup S_b$ are separated by $\ell$. Note that directly applying Kirkpatrick's algorithm does not work (see Section~\ref{sec:IOAWNN-SL} for more details). It would be more interesting to have a linear time algorithm to compute the complete diagram $\vd(S_a\cup S_b)$ by merging $\vd(S_a)$ and $\vd(S_b)$. Our technique, however, does not immediately work because it relies on the separating line $\ell$. Nevertheless, we hope our result will serve as a stepping stone towards achieving that goal. We summarize our result in the following theorem.

\begin{theorem} \label{thm:merge}
    Let $S_a$ and $S_b$ each be a set of $n$ weighted points in the plane such that all the points of $S_a\cup S_b$ are below the $x$-axis $\ell$. Given $\vd_+(S_a)$ and $\vd_+(S_b)$, $\vd_+(S_a \cup S_b)$ can be constructed in $O(n)$ time.
\end{theorem}

\paragraph{Algebraic decision tree model.} The above result holds for the standard real RAM model. Under the algebraic decision tree model in which we only count comparisons toward the time complexity, using a technique recently developed by Chan and Zheng~\cite{ref:ChanHo23}, we show that the problem IOAWNN-SL can be solved using $O(n\log n)$ comparisons. This leads to an $O(n\log n)$ time algorithm for the shortest path problem in weighted unit-disk graphs under the algebraic decision tree model, matching the $\Omega(n\log n)$ lower bound~\cite{ref:CabelloSh15}.

\paragraph{Outline.} The rest of the paper is organized as follows. We describe the shortest path algorithm framework in Section~\ref{sec:SSSP}, mainly by reviewing Wang and Xue's algorithm~\cite{ref:WangNe20}. In Section~\ref{sec:IOAWNN-SL}, we introduce our data structure for IOAWNN-SL and thus prove Theorem~\ref{thm:IOAWNN-SL}. Section~\ref{sec:merge} presents our Voronoi diagram merging algorithm for Theorem~\ref{thm:merge}. We finally describe the algebraic decision tree algorithm in Section~\ref{sec:decisiontree}.

\section{The shortest path algorithm}
\label{sec:SSSP}
In this section, we describe the shortest path algorithm. We begin with reviewing Wang and Xue's algorithm~\cite{ref:WangNe20} and explain why the IOAWNN-SL problem is a bottleneck (we only state their algorithm and refer the interested reader to their paper~\cite{ref:WangNe20} for the correctness analysis). We will show how our solution to IOAWNN-SL in Theorem \ref{thm:IOAWNN-SL} can lead to an $O(n\log^2 n/\log\log n)$ time algorithm for the shortest path problem.

Given a set $V$ of $n$ points in $\R^2$ and a source point $s \in V$, we wish to compute shortest paths from $s$ to all vertices in the weighted unit-disk graph $G = (V, E)$. We use $e_{uv} \in E$ to denote the edge between two points $u, v \in V$ and $w(e_{uv})$ to denote the weight of the edge. Recall that $w(e_{uv}) = ||u - v|| \leq 1$, where $||u - v||$ denotes the Euclidean distance between $u$ and $v$. The algorithm will compute a table $\dist[\cdot]$ such that after the algorithm finishes, $\dist[v]$ is the length of a shortest path from $s$ to $v$ for all $v\in V$. Using a predecessor table, we could also maintain a shortest path tree, but we will omit the discussion about it.

\begin{figure}
    \centering
    \includegraphics[width=1.6in]{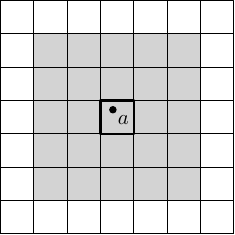}
    \caption{Illustrating $\square_a$ (the central highlighted square) and $\boxplus_a$ (the gray area). }
    \label{fig:patch}
\end{figure}

We overlay the plane with a grid $\Gamma$ of square cells with side lengths $1/2$. For any point $a\in \R^2$, denote by $\square_a$ the cell of $\Gamma$ such that $a \in \square_a$, and $\boxplus_a$ the $5 \times 5$ patch of cells in $\Gamma$ centered around $\square_a$ (see Figure~\ref{fig:patch}). For a set of points $A \subseteq \R^2$ and $a \in A$, we use $A_{\square_a} = A \cap \square_a$ and $A_{\boxplus_a} = A \cap \boxplus_a$. The algorithm makes use of the following properties: (1) For any two points $a,b$ in the same cell of $\Gamma$, $||a-b||\leq 1$ holds; (2) if $||a-b||\leq 1$, then $b$ is in $\boxplus_a$ and $a$ is in $\boxplus_b$.

Wang and Xue's algorithm is summarized in Algorithm~\ref{algo:SSSP}. It can be understood by contrasting with Dijkstra's algorithm, which we write in Algorithm \ref{algo:Dijkstra} using similar notation. In particular, a subroutine \textsc{Update}$(A, B)$ is used to ``push'' the current candidate shortest path information from $A$ to $B$ where $A, B \subseteq V$. Specifically, for each point $b \in B$, we find:
\begin{equation} \label{eq:pb}
    p_b = \argmin_{\{a \in A : e_{ab} \in E\}} \dist[a] + w(e_{ab}).
\end{equation}
We then update $\dist[b]$ to $\min\{\dist[b], \dist[p_b] + w(e_{p_bb})\}$.

\begin{algorithm}
    \DontPrintSemicolon
    \caption{Wang and Xue's algorithm \cite{ref:WangNe20}} \label{algo:SSSP}
    $\dist[a] \gets \infty$ for all $a \in V$\;
    $\dist[s] \gets 0$\;
    $A \gets V$\;
    \While{$A \neq \emptyset$}{
        $c \gets \argmin_{a \in A} \{\dist[a]\}$\;
        \textsc{Update}$(A_{\boxplus_c}, A_{\square_c})$ \tcp*{First Update} \label{ln:first}
        \textsc{Update}$(A_{\square_c}, A_{\boxplus_c})$ \tcp*{Second Update} \label{ln:second}
        $A \gets A \setminus A_{\square_c}$\;
    }
    \Return{$\dist[\cdot]$}\;
\end{algorithm}

\begin{algorithm}
    \DontPrintSemicolon
    \caption{Dijkstra's algorithm} \label{algo:Dijkstra}
    $\dist[a] \gets \infty$ for all $a \in V$\;
    $\dist[s] \gets 0$\;
    $A \gets V$\;
    \While{$A \neq \emptyset$}{
        $c \gets \argmin_{a \in A} \{\dist[a]\}$\;
        \textsc{Update}$(\{c\}, A)$\;
        $A \gets A \setminus \{c\}$\;
    }
    \Return{$\dist[\cdot]$}\;
\end{algorithm}

The main difference between Wang and Xue's algorithm and Dijkstra's is that instead of operating on single vertices, Wang and Xue's algorithm operates on cells of $\Gamma$. Generally speaking, the first update (Line~\ref{ln:first}) in Algorithm~\ref{algo:SSSP}  is to update the shortest path information for the points in $A_{\square_c}$ using the shortest path information of their neighbors. The second update is to use the shortest path information for the points in $V_{\square_c}$ to update the shortest path information of their neighbors. Wang and Xue prove that after the first update, the shortest path information for all points of $V_{\square_c}$ is correctly computed.

Wang and Xue give an $O(n \log^2 n)$ time solution for the second update, i.e., Line~\ref{ln:second}. The rest of Algorithm \ref{algo:SSSP} takes $O(n \log n)$ time. We will improve the runtime for the second update to $O(n \log^2 n / \log \log n)$ using Theorem~\ref{thm:IOAWNN-SL}, which improves the runtime for Algorithm~\ref{algo:SSSP} to $O(n \log^2 n / \log \log n)$. The details are discussed in the following.

\subsection{The second update}

To implement the second update \textsc{Update}$(A_{\square_{c}}, A_{\boxplus_{c}})$, since $A_{\boxplus_c}$ has $O(1)$ cells, it suffices to perform \textsc{Update}$(A_{\square_{c}}, A_{\square})$ for each cell $\square\in \boxplus_c$ individually.

If $\square$ is $\square_c$, then $A_{\square_c} = A_{\square}$. Since the distance between two points in $\square_c$ is at most $1$, \textsc{Update}$(A_{\square_{c}}, A_{\square})$ can be performed in $O(|A_{\square_c}|\log |A_{\square_c}|)$ time (and $O(|A_{\square_c}|)$ space) by constructing the additively-weighted Voronoi diagram for $A_{\square_c}$~\cite{ref:FortuneA87}.

If $\square$ is not $\square_c$, a useful property is that $\square$ and $\square_c$ are separated by an axis-parallel line. To perform \textsc{Update}$(A_{\square_{c}}, A_{\square})$, Wang and Xue~\cite{ref:WangNe20} proposed Algorithm~\ref{algo:update} below.

\begin{algorithm}
    \DontPrintSemicolon
    \caption{\textsc{Update}$(A, B)$ from \cite{ref:WangNe20}} \label{algo:update}
    $\dist'[a] \gets \dist[a]$ for $a \in A$\;
    Sort the points in $A = \{a_1, \ldots, a_{|A|}\}$ so that $\dist'[a_1] \leq \ldots \leq \dist'[a_{|A|}]$\;
    \For{$i = 1, \ldots, |A|$}{
        $B_i \gets \{b \in B : e_{a_i b} \in E \text{ and } e_{a_j b} \notin E \text{ for all } j < i\}$\;
    }
    $U \gets \emptyset$\;
    \For{$i = |A|, \ldots, 1$}{ \label{line:forstart}
    $U \gets U \cup \{a_i\}$\; \label{ln:insert}
    \For{$b \in B_i$}{
    $p = \argmin_{u \in U} \{\dist'[u] + w(e_{ub})\}$\; \label{ln:findp}
    $\dist[b] \gets \min\{\dist[b], \dist'[p] + w(e_{pb})\}$\;  \label{line:forend}
    }
    }
\end{algorithm}

The correctness of Algorithm \ref{algo:update} hinges on the fact that $p$ found by Line~\ref{ln:findp}  is the same as $p_b$ in Equation \eqref{eq:pb}. This is seen by arguing that $p_b \in U$ and $e_{pb} \in E$.

We now analyze the runtime of Algorithm~\ref{algo:update}. Sorting $A$ takes $O(|A|\log |A|)$ time. Computing the subsets $B_i$, $1\leq i\leq |A|$, can be done in $O((|A|+|B|)\log (|A|+|B|))$ time (and $O(|A|+|B|)$ space)~\cite{ref:WangNe20}. The for loop (Lines \ref{line:forstart}--\ref{line:forend}) is an instance of the IOAWNN-SL problem introduced in Section~\ref{sec:intro}. Indeed, if we assign each point $u$ in $U$ a weight equal to $\dist'[u]$, then $p$ in Line~\ref{ln:findp} is essentially the additively-weighted nearest neighbor of $b$ in $U$. The set $U$ is dynamically changed with point insertions in Line~\ref{ln:insert}. As such, by Theorem~\ref{thm:IOAWNN-SL}, the for loop can be implemented in $O(k\log^2 k/\log\log k)$ time (and $O(k)$ space) with $k=|A|+|B|$. Therefore, \textsc{Update}$(A_{\square_{c}}, A_{\square})$ can be performed in $O(k \log^2 k/\log\log k)$ time and $O(k)$ space, with $k=|A_{\square_c}|+|A_{\square}|$.

In summary, since $A_{\boxplus_c}$ has $O(1)$ cells, the second update \textsc{Update}$(A_{\square_{c}}, A_{\boxplus_{c}})$ can be implemented in $O(|A_{\boxplus_{c}}| \log^2 |A_{\boxplus_{c}}|/\log\log |A_{\boxplus_{c}}|)$ time and $O(|A_{\boxplus_{c}}|)$ space as $A_{\square_{c}}\subseteq A_{\boxplus_{c}}$. As analyzed in \cite{ref:WangNe20}, the total sum of $|A_{\boxplus_{c}}|$ in the entire Algorithm~\ref{algo:SSSP} is $O(n)$. This leads to the following result.

\begin{theorem} \label{thm:SSSPimprove}
    Given a set $V$ of $n$ points in the plane and a source point $s$,
    shortest paths from $s$ to all other vertices in the weighted unit-disk graph $G = (V, E)$ can be computed in $O(n \log^2 n / \log \log n)$ time and $O(n)$ space.
\end{theorem}

\section{The offline insertion-only additively-weighted nearest neighbor problem with a separating line (IOAWNN-SL)}
\label{sec:IOAWNN-SL}
In this section, we prove Theorem~\ref{thm:IOAWNN-SL}. We follow the notation in Section~\ref{sec:intro}. In particular, for any subset $P'\subseteq P$, $\vd_+(P')$ denotes the portion of the additively-weighted Voronoi diagram of $P'$ above the $x$-axis $\ell$.

Our data structure $\calD(P)$ for Theorem~\ref{thm:IOAWNN-SL} consists of two components: $\calD(P')$ and $\vd_+(P \setminus P')$ for some subset $P' \subseteq P$; we maintain the invariant $|P'| \leq |P| / \log |P|$. We also build a point location data structure on $\vd_+(P \setminus P')$ so that, given a query point, the cell of $\vd_+(P \setminus P')$ containing the point can be found in $O(\log |P \setminus P'|)$ time~\cite{ref:EdelsbrunnerOp86,ref:KirkpatrickOp83}. As such, $\calD(P)$ is a recursive structure: $\calD(P)$ is defined in terms of $\calD(P')$ which in turn is defined in terms of $\calD(P'')$ and so on. As the base case, if $|P|\leq c$ for some constant $c$, then we simply let $\calD(P)=\vd_+(P)$. Similar recursive data structures have been used before in the literature, e.g., \cite{ref:ChanDy20,ref:OvermarsTh83}.

In the following, we discuss how to handle the two operations: insertions and queries.

\paragraph{Queries.}
Given a query point $q$ above $\ell$, we first find the nearest neighbor of $q$ in $P\setminus P'$ using a point location query on $\vd_+(P\setminus P')$. Then, we find the nearest neighbor of $q$ in $P'$ using $\calD(P')$ recursively. Among the two ``candidate'' neighbors, we return the one nearer to $q$ as the answer. For the query time, since a point location query on $\vd_+(P\setminus P')$ takes $O(\log |P\setminus P'|)$ time, the query time $Q(n)$ satisfies the following recurrence: $Q(n)=Q(n/\log n)+O(\log n)$, which solves to $Q(n)=O(\log^2 n/\log\log n)$. Therefore, each query operation takes worst-case $O(\log^2 n/\log\log n)$ time.

\paragraph{Insertions.}
To insert a point $p$ below $\ell$ to $P$, we first insert $p$ to $P'$ recursively. We then check if the invariant $|P'| \leq |P| / \log |P|$ still holds. If not, we set $P' = \emptyset$, and then construct $\vd_+(P)$ as follows. First, we construct $\vd_+(P')$ recursively. Recall that $\vd_+(P\setminus P')$ is already available. We compute $\vd_+(P)$ by merging $\vd_+(P')$ and $\vd_+(P\setminus P')$, which takes $O(|P|)$ time by Theorem~\ref{thm:merge}. Finally, we construct a point location data structure on $\vd_+(P)$ in $O(|P|)$ time~\cite{ref:EdelsbrunnerOp86,ref:KirkpatrickOp83}. This finishes the insertion operation.

We now analyze the insertion time. First, suppose that we need to construct $\vd_+(P)$ due to the insertion of $p$. Then, the construction time $T(n)$ for $\vd_+(P)$ satisfies the following recurrence: $T(n)=T(n/\log n)+O(n)$, which solves to $T(n)=O(n)$.

Since $P'=\emptyset$ once $\vd_+(P)$ is constructed, we only need to construct $\vd_+(P)$ every $\Theta(n / \log n)$ insertions. As constructing $\vd_+(P)$ takes $O(|P|)$ time, the amortized time for constructing $\vd_+(P)$ per insertion is $O(\log n)$. As such, if $I(n)$ is the amortized time for each insertion, we have the following recurrence: $I(n) = I(n / \log n) + O(\log n)$. The recurrence solves to $I(n) = O(\log^2 n / \log \log n)$. We conclude that each insertion takes $O(\log^2 n / \log \log n)$ amortized time.

Note that the space $S(n)$ of $\calD(P)$ satisfies the following recurrence: $S(n)=S(n/\log n)+O(n)$, which solves to $S(n)=O(n)$. This proves Theorem \ref{thm:IOAWNN-SL}.

\section{Merging two additively-weighted Voronoi diagrams}
\label{sec:merge}
In this section, we prove Theorem \ref{thm:merge}. For completeness, we first introduce the formal definition of additively-weighted Voronoi diagrams and then present our merging algorithm.

\subsection{Additively-weighted Voronoi diagrams}
Let $S = \{s_1, s_2, \ldots, s_n\}$ be a set of $n$ points in the plane such that each point $s_i$ has a weight $w_i$ that can be positive, zero, or negative. Following the literature, we refer to points of $S$ as {\em sites}. We define the additively-weighted Euclidean distance (or {\em weighted distance} for short) of a point $p \in \R^2$ to a site $s_i$ as $d(s_i, p) = ||s_i - p|| + w_i$.

The additively-weighted Voronoi diagram of $S$, denoted by $\vd(S)$, partitions the plane into Voronoi regions, Voronoi edges, and Voronoi vertices; see Figure~\ref{fig:AWVD_DL}. Each Voronoi region $R_i$ is associated with a site $s_i$ and is defined to be the set of points that are closer to $s_i$ than to any other site measured by the weighted distances: $$R_i = \{p \in \R^2 : d(s_i, p) < d(s_j, p), \forall j \neq i\}.$$

\begin{figure}
    \centering
    \includegraphics[width=3in]{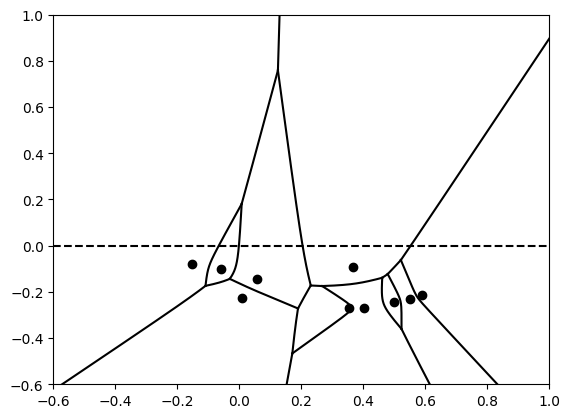}
    \caption{Illustrating an additively-weighted Voronoi diagram. The dashed horizontal line is the $x$-axis $\ell$.}
    \label{fig:AWVD_DL}
\end{figure}

Each Voronoi edge $E_{ij}$ is associated with two distinct sites $s_i$ and $s_j$ and is defined to be the set of points that are equidistant to $s_i$ and $s_j$ and closer to these sites than any other sites:
$$E_{ij} = \{p \in \R^2 : d(s_i, p) = d(s_j, p) < d(s_k, p), \forall k \neq i, j\}.$$

Each Voronoi vertex is associated with three or more distinct sites and is defined to be the point that is equidistant to these sites and closer to these sites than any other site.

We will also talk about the {\em bisector} between two sites, which is defined to be the set of points in the plane that are equidistant to the two sites:
$$B(s_i, s_j) = \{p \in \R^2 : d(s_i, p) = d(s_j, p)\}.$$
$B(s_i, s_j)$ is a hyperbolic arc whose foci are $s_i$ and $s_j$. Note that a Voronoi edge associated with two sites is a subset of their bisector.

Observation \ref{obs:basic_prop} states some properties about $\vd(S)$ that are well known in the literature; we will use these properties in our algorithm.

\begin{observation} \label{obs:basic_prop}  {\em (\cite{ref:FortuneA87})}
    \begin{enumerate}
        \item Every Voronoi region of $\vd(S)$ must contain its associated site.
        \item Each Voronoi region $R_i$ of $\vd(S)$ is star-shaped with respect to its site $s_i$, that is,
              the line segment $\overline{s_ip}$ is inside $R_i$ for any point $p\in R_i$.
        \item The combinatorial size of $\vd(S)$ is $O(|S|)$.
    \end{enumerate}
\end{observation}

\subsection{Merging algorithm for Theorem \ref{thm:merge}}

We follow the notation introduced in Section~\ref{sec:intro}, e.g., $\ell$, $n$, $S_a$, $S_b$, $\vd(S_a)$, $\vd(S_b)$, $\vd_+(S_a)$, $\vd_+(S_b)$, etc. Let $S=S_a\cup S_b$. Given $\vd_+(S_a)$ and $\vd_+(S_b)$, our goal is to compute $\vd_+(S)$ in $O(n)$ time. For ease of exposition, we make a general position assumption that no point in the plane is equidistant to four points of $S$.

\begin{figure}
    \centering
    \includegraphics[width=3in]{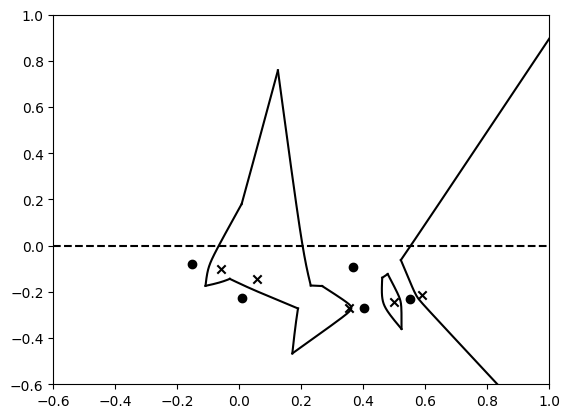}
    \caption{Illustrating the contour between two sets of points. The dashed horizontal line is the $x$-axis $\ell$.}
    \label{fig:contour}
\end{figure}

Our strategy is to identify the \emph{contour} which consists of edges in the complete Voronoi diagram $\vd(S)$ that are associated with a site in $S_a$ and a site in $S_b$; see Figure~\ref{fig:contour}. Note that the contour may have multiple connected components. The contour partitions the plane into regions $CR_i$ such that $\vd(S) \cap CR_i$ is either $\vd(S_a) \cap CR_i$ or $\vd(S_b) \cap CR_i$ (we show in Lemma~\ref{lem:contour_topology} later that each contour component has the topology of a line or a circle). As such, once we have identified the contour, computing $\vd(S)$ is straightforward. To compute the contour, the idea is to first find a point on each contour component and then trace the component by traversing $\vd(S_a)$ and $\vd(S_b)$ simultaneously. This strategy follows Kirkpatrick's algorithm~\cite{ref:KirkpatrickEf79} for merging two standard Voronoi diagrams. However, we cannot directly apply Kirkpatrick's algorithm because his method for finding a point in each contour component is not applicable to the weighted case. More specifically, his method relies on the property that the Euclidean minimum spanning tree of a point set in the plane must be a subgraph of the dual graph of its standard Voronoi diagram. However, this is not true anymore for the additively-weighted Voronoi diagrams. We make it formally as an observation below.

\begin{observation}\label{obser:example}
    The Euclidean minimum spanning tree of a set of points in the plane is not necessarily a subgraph of the dual graph of the additively-weighted Voronoi diagram of the point set.
\end{observation}
\begin{proof}
    Figure \ref{fig:counter_example} gives an example for the observation with $S=\{p_1,p_2,p_3,p_4\}$. It is obtained by setting $p_1 = (0, 4)$, $p_2 = (3, 0)$, $p_3 = (0, -4)$, and $p_4 = (-3, 0)$ with weights $w_1 = -4$, $w_2 = 0$, $w_3 = -4$, and $w_4 = 0$.
    \begin{figure}
        \centering
        \includegraphics[width=3in]{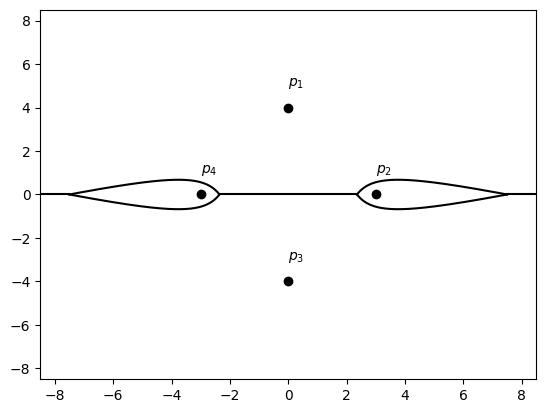}
        \caption{Illustrating the  additively-weighted Voronoi diagram of four points $\{p_1,p_2,p_3,p_4\}$ for Observation~\ref{obser:example}.}
        \label{fig:counter_example}
    \end{figure}
    Since $(p_2,p_4)$ is the closest pair among the four points of $S$, $\overline{p_2p_4}$ must be an edge in the Euclidean minimum spanning tree of $S$. However, there is no edge between $p_2$ and $p_4$ in the dual graph of the additively-weighted Voronoi diagram of $S$ because their Voronoi regions are not adjacent.
\end{proof}

In our problem, we are interested in merging $\vd_+(S_a)$ and $\vd_+(S_b)$ into $\vd_+(S)$, so it suffices to compute the portions of the contour above $\ell$. With the help of $\ell$, it is relatively easy to find a point on each contour component using the following property proved in Lemma~\ref{lem:contour_intersect_ell}: Every contour component above $\ell$ must intersect $\ell$.

At a high level, our algorithm has two main procedures. The first one is to identify the intersections between the contour and $\ell$. The second procedure is to start at these intersection points and trace each component of the contour above $\ell$.

\paragraph{The first main procedure: Finding intersections between the contour and $\ell$.}
By definition, $\ell$ is divided into segments by its intersections with $\vd_+(S_a)$, which we call {\em $\ell$-edges} of $\vd_+(S_a)$; similarly, we define {\em $\ell$-edges} for $\vd_+(S_b)$. We sweep $\ell$ from left to right, looking for places where the contour intersects $\ell$. We start with the leftmost $\ell$-edge of $\vd_+(S_a)$ and the leftmost $\ell$-edge of $\vd_+(S_b)$. At each step, we are on some $\ell$-edge $e_a$ of $\vd_+(S_a)$ and some $\ell$-edge $e_b$ of $\vd_+(S_b)$. Let $s_a\in S_a$ be the site associated with the cell of $\vd_+(S_a)$ containing $e_a$; define $s_b\in S_b$ similarly. We compute the bisector $B(s_a, s_b)$ and determine where it intersects $\ell$. The bisector is a hyperbolic arc and $\ell$ is a straight line, so they have at most two intersections $p_1$ and $p_2$. If $p_i \in e_a \cap e_b$, then $p_i$ is a point of intersection between the contour and $\ell$. In this way, we can compute all intersections between $\ell$ and the contour. Since the combinatorial sizes of $\vd_+(S_a)$ and $\vd_+(S_b)$ are $O(n)$, this procedure computes $O(n)$ intersections between $\ell$ and the contour in $O(n)$ time.

\paragraph{The second main procedure: Tracing the contour.}
We trace the contour components from the intersection points computed above. Specifically, for each intersection $p$, we trace the contour component containing $p$ as follows. Suppose that $p$ is on an $\ell$-edge $e_a$ of $\vd_+(S_a)$ and an $\ell$-edge $e_b$ of $\ell$ in $\vd_+(S_b)$. These edges are associated with sites $s_a\in S_a$ and $s_b\in S_b$. Our trace begins at $p$ and continues above $\ell$ along the bisector $B(s_a, s_b)$. This bisector enters a Voronoi region $R_a$ of $\vd_+(S_a)$ and a region $R_b$ of $\vd_+(S_b)$. We find which edge of $R_a$ or $R_b$ the bisector intersects first. If no intersection exists or the bisector first intersects $\ell$, then we finish the trace by reporting that the portion of $B(s_a, s_b)$ past $p$ is an edge of the contour. Otherwise, assume that we intersect an edge $e_a'$ of $R_a$ before an edge of $R_b$ (the case where we intersect an edge of $R_b$ first is handled the same way) and denote this point of intersection by $p'$. We rule out the case where $B(s_a, s_b)$  intersects a vertex instead of an edge because if we were to intersect a vertex, this vertex would be equidistant to three sites in $S_a$ and one site in $S_b$, which would contradict our general position assumption that no point is equidistant to four sites of $S=S_a\cup S_b$. We report that the portion of $B(s_a, s_b)$ between $p$ and $p'$ is an edge of the contour. Then, we rename $R_a$ to be the Voronoi region of $\vd_+(S_a)$ on the other side of $e_a'$ and update $p \gets p'$. We then continue the tracing from $p$ following the same process as above.

Our tracing algorithm is similar to the well-known algorithm for merging the standard Voronoi diagrams of two sets of points separated by a line~\cite{ref:ShamosCl75}. One difficulty with our algorithm is efficiently determining which edges of $R_a$ and $R_b$ the contour intersects first. This may not be a constant time operation since $R_a$ and $R_b$ may have many edges. The merge algorithm by Shamos and Hoey~\cite{ref:ShamosCl75} takes advantage of the fact that the contour in their problem is monotone so that they can find all contour edges in a region by a single scan of the boundary of that region. In our problem, the contour may not be monotone. To resolve the issue, we follow the same technique used by Kirkpatrick \cite{ref:KirkpatrickEf79} for merging standard Voronoi diagrams of two arbitrary sets of points. Specifically, before our tracing algorithm, we subdivide Voronoi regions of $\vd_+(S_a)$ and $\vd_+(S_b)$ each into sub-regions of at most four edges by drawing segments between each site and each vertex of the Voronoi region of the site (see Figure \ref{fig:spokes}; we can do this because each Voronoi region is star-shaped by Observation~\ref{obs:basic_prop}); as in \cite{ref:KirkpatrickEf79}, we refer to these segments as {\em spokes}. Because each sub-region only has at most four edges, finding where a bisector intersects a sub-region can be done in $O(1)$ time. We then apply our above tracing algorithm using these subdivisions of $\vd_+(S_a)$ and $\vd_+(S_b)$. Each tracing step now finds an intersection between the contour and either a spoke or a Voronoi edge in constant time. As such, the total time of the tracing procedure is linear in the number of such intersections. By Lemma~\ref{lem:num_contour_intersections}, the number of such intersections, and hence the runtime of the tracing procedure, is $O(n)$. Therefore, the total time of the algorithm for merging $\vd_+(S_a)$ and $\vd_+(S_b)$ is $O(n)$. This proves Theorem \ref{thm:merge}.

\begin{figure}
    \centering
    \includegraphics[width=3in]{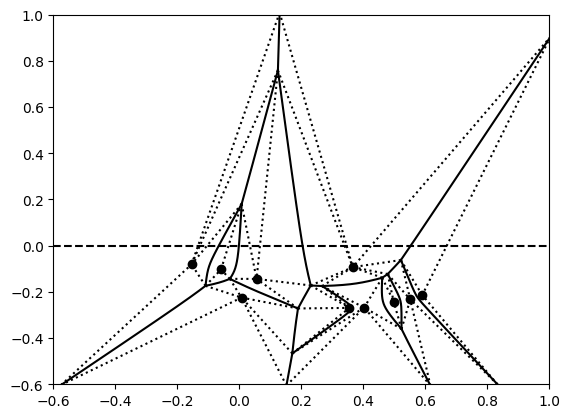}
    \caption{The dotted segments are spokes. Our algorithm only uses the portions of these spokes above $\ell$, the dashed line.}
    \label{fig:spokes}
\end{figure}

\subsection{Useful lemmas}

It remains to prove Lemmas~\ref{lem:contour_topology}, \ref{lem:contour_intersect_ell} and \ref{lem:num_contour_intersections}, which our algorithm relies on.

Recall that the contour also includes its portions below $\ell$, i.e., it is defined with respect to the complete Voronoi diagram $\vd(S)$. We first have the following lemma, which is also needed in the proof of Lemma~\ref{lem:contour_intersect_ell}; a similar result on the standard Voronoi diagrams is already used in \cite{ref:KirkpatrickEf79}.

\begin{lemma} \label{lem:contour_topology}
    Each contour component never terminates or splits; that is, it has the topology of an infinite line or a circle.
\end{lemma}
\begin{proof}
    A contour component is made up of edges in $\vd(S)$, so if it were to terminate or split, it would be at a Voronoi vertex of $\vd(S)$. Due to our general position assumption that no point is equidistant to four sites, each Voronoi vertex in $\vd(S)$ is adjacent to three sites in $S$. If the contour hits a Voronoi vertex $v$, then at least one of these sites must be in $S_a$ and at least one must be in $S_b$. Without loss of generality, let these sites be $s_1$, $s_2$, and $s_3$ with $s_1, s_2 \in S_a$ and $s_3 \in S_b$. The Voronoi edge between $s_1$ and $s_3$ and the Voronoi edge between $s_2$ and $s_3$ will be on the contour, so the contour will not terminate at $v$. The edge between $s_1$ and $s_2$ will not be on the contour, so the contour will not split at $v$.
\end{proof}

\begin{lemma} \label{lem:contour_intersect_ell}
    If a contour component contains a point above $\ell$, then the contour component must intersect $\ell$.
\end{lemma}
\begin{proof}
    Lemma \ref{lem:contour_topology} establishes that a contour component divides the plane into two regions, called {\em contour regions}. Notice that because a contour component is made up of edges in $\vd(S)$, each contour region must contain at least one Voronoi region of $\vd(S)$ and thus contains at least one site of $S$ by Observation \ref{obs:basic_prop}.

    Now assume to the contrary that a contour component $C$ contains a point above $\ell$ but $C$ does not intersect $\ell$. Then, the entire $C$ is above $\ell$. As such, one of the contour regions divided by $C$ must be entirely in the halfplane above $\ell$; let $R$ be the region. This implies that the sites of $S$ contained in $R$ must be above $\ell$, but this contradicts the fact that all sites of $S$ are below $\ell$.
\end{proof}

\begin{lemma} \label{lem:num_contour_intersections}
    \begin{enumerate}
        \item The total number of intersections between the contour and the Voronoi edges of the complete Voronoi diagrams $\vd(S_a)$ and $\vd(S_b)$ is at most $O(n)$.
        \item The total number of intersections between the contour and the spokes of the complete Voronoi diagrams $\vd(S_a)$ and $\vd(S_b)$ is at most $O(n)$.
    \end{enumerate}
\end{lemma}
\begin{proof}
    We adapt the proof from \cite{ref:KirkpatrickEf79} for a similar lemma on standard Voronoi diagrams.

    Notice that the intersection between the contour and a Voronoi edge in $\vd(S_a)$ or $\vd(S_b)$ is a vertex in $\vd(S)$. There are $O(n)$ vertices in $\vd(S)$, so the total number of intersections between the contour and the Voronoi edges of $\vd(S_a)$ and $\vd(S_b)$ is at most $O(n)$. This proves the first lemma statement.

    To prove the second lemma statement, we show that the contour can intersect each spoke at most once. We exploit the fact that Voronoi regions are star-shaped (Observation \ref{obs:basic_prop}). If the contour intersects a spoke of the Voronoi region for site $s$ in $\vd(S_a)$ or $\vd(S_b)$, then the open segment between $s$ and this intersection will lie in the Voronoi region of $s$ in $\vd(S)$. Because this segment is in the Voronoi region for $s$ in $\vd(S)$, the contour cannot intersect this segment.

    Now, assume for the sake of contradiction that the contour were to intersect a spoke twice. This would mean the closer to $s$ of the two intersections would lie on the segment between $s$ and the further of the two intersections, which we have shown above to be impossible. Therefore, the contour can only intersect each spoke at most once, and there are $O(n)$ spokes in $\vd(S_a)$ and $\vd(S_b)$, so the total number of intersections between the contour and the spokes is at most $O(n)$.
\end{proof}

\section{Algebraic decision tree algorithm}
\label{sec:decisiontree}

Under the algebraic decision tree model, where the time complexity is measured only by the number of comparisons, we show that the IOAWNN-SL problem can be solved using $O(n\log n)$ comparisons. Consequently, we can solve the shortest path problem in weighted unit-disk graphs in $O(n\log n)$ time under the algebraic decision tree model. In the following, we first describe an $O(n\log^2 n)$ time algorithm under the conventional computational model and then show how to improve it to $O(n\log n)$ time under the algebraic decision tree model.

Let $p_1,p_2,\ldots,p_n$ be the points to be inserted in this order; each point has a weight. Let $P$ denote the set of all these points. Let $Q$ be a set of $O(n)$ query points, such that all points of $P$ are above the $x$-axis $\ell$ while all points of $Q$ are below $\ell$. For each query point $q\in Q$, we know the timer when the query is conducted, i.e., we know the index $i$ such that the query looks for the nearest neighbor of $q$ among the first $i$ points of $P$. Our goal is to answer all queries for the points of $Q$.

We construct a complete binary tree $T$ whose leaves from left to right correspond to points $p_1,p_2,\ldots,p_n$ in this order. For each node $v\in T$, let $P_v$ denote the set of points that are in the leaves of the subtree rooted at $v$. Let $\vd(P_v)$ be the additively-weighted Voronoi diagram for the weighted points of $P_v$; let $\vd_+(P_v)$ be the portion of $\vd(P_v)$ above $\ell$. We construct $\vd_+(P_v)$. If we construct $\vd_+(P_v)$ for all nodes $v$ of $T$ in a bottom-up manner and use our linear time merge algorithm in Theorem~\ref{thm:merge}, constructing the diagrams $\vd_+(P_v)$ for all nodes $v\in T$ can be done in $O(n\log n)$ time. In addition, we construct Kirkpatrick's point location data structure~\cite{ref:KirkpatrickOp83} on $\vd_+(P_v)$ for each node $v\in T$, which takes $O(|P_v|)$ time. \footnote{Note that Kirkpatrick's data structure is originally for planar subdivisions in which each edge is a straight line segment. However, as discussed in~\cite{ref:KirkpatrickOp83}, the algorithm also works for additively weighted Voronoi diagrams (and other types of Voronoi diagrams) since each cell of the diagram is star-shaped. A subtle issue in our problem is that $\vd_+(P_v)$ is only the portion of the complete diagram $\vd(P_v)$ above $\ell$, and each cell of $\vd_+(P_v)$ does not contain its site. To circumvent the issue, we can enlarge each cell of $\vd_+(P_v)$ by including its site, as follows. For each cell $R\in \vd_+(P_v)$, if $\overline{ab}$ is a maximal segment of $R\cap \ell$, then we add the triangle $\triangle pab$ to $R$, where $p$ is the site of $R$. Note that $\triangle pab$ must be inside the cell of $p$ in $\vd(P_v)$, denoted by $R'$. As such, the enlarged region $R$ is still star-shaped, contains its site $p$, and is a subset of $R'$. We can then construct Kirkpatrick's point location data structure on the subdivision of all these enlarged regions $R$.} Note that we use Kirkpatrick's point location data structure instead of others such as the one in~\cite{ref:EdelsbrunnerOp86} because we will need to apply a technique from \cite{ref:ChanHo23} that requires Kirkpatrick's data structure. Constructing the point location data structures for all nodes of $T$ takes $O(n\log n)$ time.

Consider a query point $q\in Q$. Suppose we are looking for the nearest neighbor of $q$ among the first $i$ points $p_1,p_2,\ldots,p_i$ of $P$. Let $v_i$ be the leaf of $T$ corresponding to $p_i$. Following the path in $T$ from the root to $v_i$, we can find a set $V_q$ of nodes of $T$ such that the union of $P_v$ for all $v\in V_q$ is exactly $\{p_1,p_2,\ldots,p_i\}$. As such, the query can be answered after performing $O(\log n)$ point location queries on $\vd_+(P_v)$ for all $v\in V_q$. As each point location query takes $O(\log n)$ time, answering the nearest neighbor query for $q$ can be done in $O(\log^2 n)$ time. Therefore, the total time for answering the queries for all points of $Q$ is $O(n\log^2 n)$.

The above solves the problem in $O(n\log^2 n)$ time. To improve the time to $O(n\log n)$, the bottleneck is to solve all $O(n\log n)$ point location queries. For this, we resort to a technique recently developed by Chan and Zheng~\cite{ref:ChanHo23} under the algebraic decision tree model. We can simply apply \cite[Theorem~7.2]{ref:ChanHo23} to solve all our point location queries using $O(n\log n)$ comparisons (specifically, following the notation in \cite[Theorem~7.2]{ref:ChanHo23}, we have $t=O(n)$, $L=O(n\log n)$, $M=O(n\log n)$, and $N=O(n)$ in our problem; according to the theorem, all point location queries can be solved using $O(L+M+N\log N)$ comparisons, which is $O(n\log n)$). Note that the theorem statement requires the input planar subdivisions to be triangulated. The triangulation is mainly used to construct Kirkpatrick's point location data structure~\cite{ref:KirkpatrickOp83} on each planar subdivision. Since we already have Kirkpatrick's point location data structure for each $\vd_+(P_v)$ as discussed above, we can simply follow the same algorithm of the theorem.

\end{document}